\documentclass[10pt]{article}
\usepackage[usenames,dvipsnames]{color}
\usepackage{amsmath,amsthm,amsfonts,amssymb,multicol,amscd,amsbsy}
\usepackage{graphicx}
\usepackage{booktabs}
\usepackage{subfigure}
\usepackage{enumerate}
\usepackage{xcolor}
\usepackage{url}
\usepackage[nodayofweek]{datetime}
\usepackage[english]{babel}
\usepackage[toc,page]{appendix}
\usepackage{verbatim} 
\usepackage{amsthm} 
\usepackage{enumitem}
\usepackage{enumerate}
\usepackage{bbm} 
\usepackage[normalem]{ulem} 
\usepackage{bm}
\usepackage{hyperref}
\usepackage{amsmath}

\usepackage{extarrows}
\usepackage[normalem]{ulem}

\usepackage{pgfplots}
\usepackage{caption}

\usepackage[top=2.75cm,bottom=3.75cm]{geometry}

\usepackage{authblk}

\newcommand{\beq}{\begin{equation}}
\newcommand{\eeq}{\end{equation}}
\newcommand{\beqs}{\begin{equation*}}
\newcommand{\eeqs}{\end{equation*}}
\newcommand{\bal}{\begin{align}}
\newcommand{\eal}{\end{align}}
\newcommand{\bals}{\begin{align*}}
\newcommand{\eals}{\end{align*}}

\newcommand{\x}{\mathbf{x}}

\newcommand{\set}[1]{\mathbb{#1}}

\newcommand{\R}{\set{R}}

\newcommand{\Z}{\set{Z}}

\newcommand{\abs}[1]{\left\vert#1\right\vert}

\def\1{{\mathsf 1}}

\newtheorem{thm}{Theorem}
\newtheorem{lem}{Lemma}

\theoremstyle{definition}

\numberwithin{equation}{section}

\theoremstyle{remark}
\newtheorem{remark}{Remark}

\def\dotuline{\bgroup
  \ifdim\ULdepth=\maxdimen  
   \settodepth\ULdepth{(j}\advance\ULdepth.4pt\fi
  \markoverwith{\begingroup
  \advance\ULdepth0.08ex
  \lower\ULdepth\hbox{\kern.15em .\kern.1em}%
  \endgroup}\ULon}

\def\dashuline{\bgroup
  \ifdim\ULdepth=\maxdimen  
   \settodepth\ULdepth{(j}\advance\ULdepth.4pt\fi
  \markoverwith{\kern.15em
  \vtop{\kern\ULdepth \hrule width .3em}%
  \kern.15em}\ULon}
\allowdisplaybreaks

\usepackage{authblk}

\pgfplotsset{compat=1.14}
\begin{document}

\setlength{\columnsep}{5pt}

\title{Landscape Theory for Schr\"odinger Operators with General Hopping Terms on a Finite Lattice}
\author[1]{John Buhl, Isaac Cinzori, Isabella Ginnett, Mark Landry\footnote{Corresponding author:\href{mailto: landrym5@msu.edu}{ landrym5@msu.edu}}, Yikang Li, Xingyan Liu \\ of Michigan State University\footnote{This research was partially supported by NSF DMS-1758326 and DMS-1846114 grants.}}

\date{}
\maketitle

\begin{abstract}
    Findings by M. L. Lyra, S. Mayboroda and M. Filoche relate invertibility and positivity of a class of discrete Schr\"odinger matrices with the existence of the ``Landscape Function,'' which provides an upper bound on all eigenvectors simultaneously. Their argument is based on the variational principles. We consider an alternative method of proving these results, based on the power series expansion of matrices, and demonstrate that it naturally extends the original findings to the case of long range operators.
\end{abstract}

\section{Introduction}
The findings of Lyra, Mayboroda, and Filoche \cite{LMF15} that we are concerned with are summarized below. Their paper examined the following Schr\"odinger matrix $H$, a matrix form of the equation $[-\Delta + V(\vec{x})]\Psi(\vec{x}) = E \Psi(\vec{x})$, on a finite, discrete lattice with a hopping distance of one:
\begin{align}\label{eq:land}
   H=\begin{pmatrix}
	v_1 & -1 & 0  & \cdots &  0 \\
	-1 & v_2 & \ddots &\ddots   &  \vdots \\
	0 & \ddots& \ddots & \ddots  & 0 \\
	\vdots  & \ddots & \ddots & v_{n-1} & -1 \\
	0  & \cdots & 0 & -1 & v_n
\end{pmatrix} \ \ v_j \ge 2, \ j=1,\cdots,n.
\end{align}
It was shown in the original paper \cite{LMF15} that there exists a unique solution $\vec{u}$ to the ``Landscape Function'' equation
\begin{align}
\label{landscape_equation}
\ H\vec{u}=\begin{pmatrix} 1\\\vdots\\1\end{pmatrix}.
\end{align}
The solution vector $\vec{u}$ is also called the ``Landscape Function," and was shown to have the property that for any eigenvector $\vec{x}$ and corresponding eigenvalue $\lambda$ of the matrix, so that $H\vec{x} = \lambda\vec{x}$, and the above conditions on $v_j$, we have, for all $j=1,\cdots,n$
\begin{align*}
\frac{|x_j|}{\max\limits_{1\le k\le n}|x_k|}\le \lambda u_j
\end{align*}
The above result is what we will call throughout the ``Landscape---after what is in physical practice a landscape, $\vec{u}$---Theory." In applied situations, the above result shows that every eigenmode of the given physical system is bounded in some way by the ``Landscape Function" $\vec{u}$. Results of this kind are strongly related to Anderson localization: see, for example, \cite{WZ} in the case of the nearest neighbor hopping. We were motivated by these initial results to devise an extension to more complicated (long range) versions of the Schr\"odinger matrix,  such as the next nearest neighbor hopping in the extended Harper model. For more information, see \cite{AJM} and references therein.

Before we move to the results, we want to make a few more comments about the original Landscape Theory. Anderson localization \cite{A58} in a disordered medium is one of the most important and popular topics in condensed matter physics. A new concept, the \emph{landscape function} for an elliptic differential operator $L$, was first introduced in 2012 by Filoche and Mayboroda \cite{FM12}, and was shown to be extremely adept at predicting the location of regions of low energy eigenstates of $L$. The concept of the \emph{landscape function} was generalized from the continuous case to the discrete case in \cite{LMF15}, which our current paper is based on. The Landscape Theory was further developed in mathematics \cite{S17,ADFJM}, as well as in theoretical and experimental physics \cite{F+17,L+16}. We refer readers to the above papers and the references therein for more background and details of the Landscape Theory.

We present our results in two sections. The first covers an extension of a slightly weakened version of the initial Landscape Theory to a ``Long Range" Schr\"odinger matrix. We then provide the findings necessary to strengthen our theory so that it is a complete extension of the original results. It must also be highlighted that our extensions are built off a new method of proof of the Landscape Theory --- which also holds for the original model with hopping distance one. This method is the power series expansion of matrices. The rest of the paper is devoted to extending and proving our results.

\section{Strict Hamiltonian Potential Inequality}\label{sec:strict}
Let the Long Range Schr\"odinger matrix be given by $H=-H_0+V$, where, for 
$n\in\mathbb{N}$ and $n\ge 2$,
\begin{align}\label{eq:V}
   V=\begin{pmatrix}
	v_1 & 0 & 0  & \cdots &  0 \\
	0 & v_2 & \ddots &\ddots   &  \vdots \\
	0 & \ddots& \ddots & \ddots  & 0 \\
	\vdots  & \ddots & \ddots & v_{n-1} & 0 \\
	0  & \cdots & 0 & 0 & v_n
\end{pmatrix}
\end{align}
and
\begin{align}\label{eq:H0}
   H_0=\begin{pmatrix}
	0 & a_1 & a_2 & \cdots &   \cdots &  a_{n-2} & a_{n-1} \\
	a_1 & 0 & a_1 & \ddots & \ddots & \ddots & a_{n-2} \\
	a_2  & a_1 & 0 & \ddots & \ddots  & \ddots & \vdots  \\
	\vdots & \ddots & \ddots & \ddots& \ddots&\ddots&\vdots \\
	\vdots & \ddots& \ddots & \ddots & 0  &  a_1 & a_2 \\
	a_{n-2}  &\ddots &  \ddots  & \ddots & a_1 &  0 & a_1 \\ 
	a_{n-1}& a_{n-2} & \cdots  & \cdots & a_2 & a_1 & 0
\end{pmatrix}, \ a_i\ge0, \ i=1,\cdots,n-1.
\end{align}

\begin{thm}[Landscape Theory for general hopping matrix on a finite lattice]\label{thm:genLST}
Let $H$ be given as above. We consider the eigenvector $\vec{x}$ of $H$ with corresponding eigenvalue $\lambda$, so that $H\vec{x}=\lambda \vec{x}$. Assume that
\begin{align}\label{eq:cond-v}
    v_j>2\sum_{i=1}^{n-1}a_i, \ \ j=1,\cdots, n
\end{align}
Then, the solution to the Landscape Function equation, $\vec u\in \R^n$, exists, satisfying $H\vec{u}=\vec{1}$.\footnote{We use the notation $\vec1=(1,\cdots,1)^T$.} Further, for all $j=1,\cdots,n$
\begin{align}\label{eq:LST}
   \frac{|x_j|}{\max\limits_{1\le k\le n}|x_k|}\le \lambda u_j, 
\end{align}
\end{thm}

\begin{remark}\label{rmk:H0}
If $a_1=1$ and $a_2=\cdots=a_{n-1}=0$, then Theorem \ref{thm:genLST} under condition \eqref{eq:cond-v} gives a weakened version of the original landscape theory. It is also easy to check that if $a_1=\cdots=a_{n-1}=0$, Theorem \ref{thm:genLST} holds trivially for the diagonal matrix $H=V$. The stronger version of Theorem \ref{thm:genLST}, with a soft inequality in \eqref{eq:cond-v}, also holds true. We will discuss that in Section \ref{sec:soft}. 
\end{remark} 

The outline of the proof will be similar to the original Landscape Theory: (i) the existence of the inverse and the Landscape Function solution; (ii) the positivity of the inverse and Landscape Function. However, we use an alternative proof, built off the power series expansion of the pertinent matrices, to show these results.

\begin{lem}\label{lem:inverse}
 If $ v_j$ satisfies \eqref{eq:cond-v}, then $H$ is invertible. As a consequence, there is always a vector $\vec u\in \R^n$ satisfying $H\vec u=\vec1$, with the explicit expression
\begin{align}
\label{landveccomp}
    u_j=\sum_{k=1}^nG_{jk},
\end{align}
  where $G_{ij}=H^{-1}(i,j)$ is the $(i,j)th$ entry of the inverse of $H$.
\end{lem}
\begin{remark}\label{remark:evalue}
Moreover, all eigenvalues of $H$ are strictly positive. This is easily deduced from the fact that the matrix $H$ is both self-adjoint and strictly diagonally dominant by construction. Together, this means, due in part to positive semi-definiteness, that all eigenvalues of $H$ are real and greater than or equal to zero. Further, as $H$ is invertible, zero cannot be an eigenvalue and the result follows.
\end{remark}

\begin{lem}\label{lem:posi-Green}
For $i,j=1,\cdots,n$, $$G_{ij}\ge0,$$ furthermore,
\begin{align*}
    u_j>0. 
\end{align*}
\end{lem}

We will prove Lemma \ref{lem:inverse} and Lemma \ref{lem:posi-Green} later. We will first complete the proof of Theorem \ref{thm:genLST}.

\begin{proof}[Proof of Theorem \ref{thm:genLST}]
Let $\vec{x}$ be an eigenvector of $H$ with eigenvalue $\lambda$:
\begin{equation*}
H\vec{x}  =  \lambda \vec{x}
\end{equation*}
\begin{equation}
\label{eigenvec}
\vec{x}  =  \lambda H^{-1} \vec{x}
\end{equation}
\begin{equation}
\label{eigenveccomp}
x_j  =  \lambda \sum_{k=1}^{n} G_{jk} x_k,
\end{equation}
where \eqref{eigenvec} follows from Lemma \ref{lem:inverse} and \eqref{eigenveccomp} follows from Lemma \ref{lem:inverse} and matrix multiplication. We now examine the vector $\vec{x}$ scaled by its maximum value, entry-wise. Using Remark \ref{remark:evalue}, we can assume $\lambda>0$: 
\begin{eqnarray}
\nonumber
\frac{|x_j|}{\max\limits_{1\le k\le n}|x_k|} & = & \lambda\ \left|\sum_{k=1}^{n} G_{jk} \frac{x_k}{\max\limits_{1\le k\le n}|x_k|}\right|\
\\
\label{bytrieq}
 & \le & \lambda \sum_{k=1}^{n} \left|G_{jk} \frac{x_k}{\max\limits_{1\le k\le n}|x_k|}\right| 
\\
\label{bylem5}
 & \le & \lambda \sum_{k=1}^{n} G_{jk} \left|\frac{x_k}{\max\limits_{1\le k\le n}|x_k|}\right| 
\\
\label{upboundone}
 & \le & \lambda \sum_{k=1}^{n} G_{jk} 
\\
\label{bylem3rem4}
 & = & \lambda  u_j,
\end{eqnarray}
where \eqref{bytrieq} follows by the triangle inequality, \eqref{bylem5} by Lemma \ref{lem:posi-Green}, \eqref{upboundone} by our upper bound being one, and \eqref{bylem3rem4} by Lemma \ref{lem:inverse}.
\end{proof}

Our proofs for Lemma \ref{lem:inverse} and Lemma \ref{lem:posi-Green} rely on the following result which pertains to power series expansion of matrices.
\begin{lem}\label{lem:power}
For two $n\times n$ matrices $A$ and $B$, if $A$ is invertible and $\|A^{-1}B\|<1$,\footnote{\;$\|\cdot\|$ denotes the operator norm.} then $A-B$ is invertible and 
\begin{align}\label{eq:power1}
    (A-B)^{-1}=\sum_{k=0}^{\infty}\left(A^{-1}B\right)^kA^{-1}.
\end{align}
\end{lem}
Lemma \ref{lem:power} is quite standard. For the sake of completeness, we include the proof here for the readers' convenience. 
\begin{proof}[Proof of Lemma \ref{lem:power}]
It is a known theorem in Analysis that, provided $\|A\| < 1$ and letting $I$ denote the identity matrix
\begin{align}\label{eq:power2}
    (I - A)^{-1}=\sum_{k=0}^{\infty}A^k.
\end{align}
Thus, we can see from the following manipulations that
\begin{align*}
(A - B)^{-1} = (A(I - A^{-1}B))^{-1} = (I - A^{-1}B)^{-1}A^{-1}. 
\end{align*}
Hence, provided $\|A^{-1}B\| < 1$,
\begin{align*}
    (I - A^{-1}B)^{-1}A^{-1} = \sum_{k=0}^{\infty}(A^{-1}B)^kA^{-1}.
\end{align*}
\end{proof}

 In order to apply Lemma \ref{lem:power}, we also need estimates of the matrix norms of $V$ and $H_0$. We can prove that
\begin{lem}\label{lem:norm}
\begin{align}\label{eq:Vnorm}
    \|V^{-1}\|\le \max_{1\le j\le n} v_j^{-1}
\end{align}
and 
\begin{align}\label{eq:H0norm}
    \|H_0\|\le 2\sum_{i=1}^{n-1}a_i
\end{align}
\end{lem}
\begin{proof}
We will start with (\ref{eq:Vnorm}). It is observed that $V$ is a diagonal matrix with all diagonal entries strictly greater than zero (in both Theorems \ref{thm:genLST} and \ref{thm:genLST2}, which we will discuss later). Thus we can safely say that $V^{-1}$ is well defined and has diagonal entries $(v_j)^{-1} = 1/v_j$. Further, as $V^{-1}$ is a diagonal, self-adjoint matrix, we know that $\|V^{-1}\|$ is equal to its maximal eigenvalue, which in this case is simply $\max\limits_{1\le j\le n}  v_j^{-1}$.\footnote{\,This relationship between a non-negative self-adjoint matrix's operator norm and its maximal eigenvalue is a known---albeit more advanced---result in Linear Algebra}
\end{proof}
\begin{proof}
Now we move on to \eqref{eq:H0norm}. We start with a definition of the notation $R_k$ and $L_k$ as the matrices with ones on the $k^{th}$ right and left off-diagonals, respectively, and zero elsewhere. Now we can see that $H_0$, as given in initial equation \eqref{eq:H0}, has the following decomposition
\begin{align*}
    H_0 = a_1R_1 + a_1L_1 + \cdots + a_{n-1}R_{n-1} + a_{n-1}L_{n-1}.
\end{align*}
Thus the norm of $H_0$ obeys the following statement, based of the triangle inequality and properties of norms
\begin{eqnarray*}
    \|H_0\| & = & \|a_1R_1 + a_1L_1 + \cdots + a_{n-1}R_{n-1} + a_{n-1}L_{n-1}\|
\\
    \|H_0\| & \le & \abs{a_1}\|R_1\| + \abs{a_1}\|L_1\| + \cdots + \abs{a_{n-1}}\|R_{n-1}\| + \abs{a_{n-1}}\|L_{n-1}\|
\end{eqnarray*}

Further, we can show that all norms $\|R_i\|$ and $\|L_i\|, \ i = 1,\cdots,n-1$ are less than or equal to one. This is demonstrated in several steps. First, it can be shown by direct computation that $R_i = (R_1)^i$ and $L_i = (L_1)^i$. Then, one can see
\begin{align*}
    \|R_i\| =  \|(R_1)^i\| \le \|R_1\|^i 
\end{align*}
The same results apply to all $L_i$. The final step to reach our assertion that all the pertinent norms are less than or equal to one is to show that $\|R_1\|$ and $\|L_1\|$ are less than or equal to one, which is done below.
    \begin{align*}
    L_1:=L=\begin{pmatrix}
	0 & 0 & 0 & \cdots & 0 \\
	1 & 0 & 0 & \ddots & \vdots \\
	0 & \ddots & \ddots & \ddots & 0 \\
	\vdots & \ddots & \ddots & 0 & 0 \\
	0 & \cdots & 0 & 1 & 0
\end{pmatrix}, \  \ 
R_1:=R=\begin{pmatrix}
	0 & 1 & 0 & \cdots & 0 \\
	0 & 0 & 1 & \ddots & \vdots \\
	0 & \ddots & \ddots & \ddots & 0 \\
	\vdots & \ddots & \ddots & 0 & 1 \\
	0 & \cdots & 0 & 0 & 0
\end{pmatrix}.
\end{align*}
For any  $\vec{x}=(x_1,x_2,\cdots,x_n)^T$, direct computation shows that 
\begin{align*}
   R\vec{x}=R\begin{pmatrix}x_1\\x_2\\x_3\\
   \vdots\\ x_{n-1}\\x_n\end{pmatrix}=\begin{pmatrix}0\\x_1\\x_2\\
   \vdots\\x_{n-2}\\ x_{n-1}\end{pmatrix}
\end{align*}
Therefore, 
\begin{align} \label{eq:Rx}
    \|R\vec{x}\|^2=x_1^2+x_2^2+\cdots+x_{n-1}^2\le 
    x_1^2+x_2^2+\cdots+x_{n-1}^2+x_n^2=\|\vec{x}\|^2
\end{align}
which implies $\|R\vec{x}\|\le \|\vec{x}\|$. 
According to the definition of the matrix operator norm and \eqref{eq:Rx}
\begin{align*}
    \|R\|=\max_{\vec{x}\neq \vec{0}}\frac{\|R\vec{x}\|}{\|\vec{x}\|}\le 1.
\end{align*}
Exactly the same argument shows that $\|L\|\le 1$, which completes the assertion. Putting all of the above findings together, we can substitute all $\|R_i\|$s and $\|L_i\|$s in our $H_0$ equation with ones via an inequality.
\begin{align*}
    \|H_0\| \le \sum_{i=1}^{n-1}\abs{a_i} + \sum_{i=1}^{n-1}\abs{a_i} = 2\sum_{i=1}^{n-1}a_i
\end{align*}
This is what we desire. The absolute value falls away because all $a_i$s are positive by construction.
\end{proof}
Now we can proceed to prove Lemma \ref{lem:inverse} and Lemma \ref{lem:posi-Green} using equation \eqref{eq:power1} and Lemma \ref{lem:norm}. 

\begin{proof}[Proof of Lemma \ref{lem:inverse}]
Given the structure of $V$ and $H_0$ given in Theorem \ref{thm:genLST}, in particular \eqref{eq:cond-v}, we can see that by Lemma \ref{lem:norm}
\begin{align*}
    \|V^{-1}\|\le \max_{1\le j\le n} v_j^{-1} < \frac{1}{2\sum_{i=1}^{n-1}a_i}
\end{align*}
and, as proven
\begin{align*}
    \|H_0\|\le 2\sum_{i=1}^{n-1}a_i.
\end{align*}
Thus, we see that
\begin{align*}
    \|V^{-1}H_0\| \le \|V^{-1}\| \|H_0\| < \frac{2\sum_{i=1}^{n-1}a_i}{2\sum_{i=1}^{n-1}a_i}=1.
\end{align*}
Hence, we meet the conditions to satisfy Lemma \ref{lem:power} and have invertibility of our matrix $H$.
\end{proof}

\begin{proof}[Proof of Lemma \ref{lem:posi-Green}]
With the same conditions as for the proof of Lemma \ref{lem:inverse}, we see that every entry of $V^{-1}$ --- whose diagonal elements $v_j^{-1}, \ j = 1,\cdots,n,$ are the reciprocals of the diagonal elements of $V$ --- and every entry of $H_0$ are non-negative, based on the structure of these two matrices as given in Theorem \ref{thm:genLST}. We thus have that every entry of $(V-H_0)^{-1}$ is non-negative by the structure of the power series expansion given in Lemma \ref{lem:power}. Next, since $V-H_0$ is invertible, we have that none of the rows of $(V-H_0)^{-1}$ are identically zero. Together with \eqref{landveccomp}, this implies $u_j>0$ for $j=1,\cdots,n.$
\end{proof}

\section{Soft Hamiltonian Potential Inequality}\label{sec:soft}
In this section, we would like to study the optimality of condition \eqref{eq:cond-v}.  We now prove that
\begin{thm}\label{thm:genLST2}
Let $H$ be given as in the Theorem \ref{thm:genLST}. All conclusions of Theorem \ref{thm:genLST} hold true under the condition\footnote{Clearly, as discussed in Remark \ref{rmk:H0}, if $a_1=\cdots=a_{n-1}=0$, we still need the strict inequality $v_j>0$ for the trivial diagonal case to be true.}
\begin{align}\label{eq:cond-v2}
    v_j\ge 2\sum_{i=1}^{n-1}a_i, \ \ j=1,\cdots, n.
\end{align}
\end{thm}
As we see from the method of proof employed in Section \ref{sec:strict}, it is enough to show
\begin{lem}\label{lem:hardH0}
\begin{align}\label{eq:H0norm2}
    \|H_0\|< 2\sum_{i=1}^{n-1}a_i.
\end{align}
\end{lem}
It follows from
\begin{lem}\label{lem:offdiagonal}
If $A_j(n)$ is an $n \times n$ matrix given by
\begin{align}\label{eq:Ajn}
   A_j(n)=\begin{pmatrix}
	0 & 0 & \cdots & 1 & 0 & \cdots & 0 \\
	0 & 0 & \ddots &\ddots & \ddots &  \ddots & \vdots \\
	\vdots & \ddots & \ddots & \ddots & \ddots & \ddots & 0 \\
	1 & \ddots & \ddots & \ddots & \ddots & \ddots & 1 \\
	0 & \ddots & \ddots & \ddots & \ddots & \ddots & \vdots \\
	\vdots & \ddots & \ddots & \ddots & \ddots & 0 & 0 \\
	0  &\cdots & 0& 1 & \cdots & 0 & 0
\end{pmatrix} =L_j+R_j, 
\end{align}
then all the eigenvalues of $A_j(n)$ are strictly between $-2$ and $2$. As a consequence, due to the fact that $A_j(n)$ is self-adjoint, we have $\|A_j\|<2$.
\end{lem}
\begin{remark}
If $n=1$, we allow $j=1$, and denote the trivial case $A_1(1)=(0)$. 
\end{remark}

\begin{lem}\label{lem:basis}
 Let ${\mathcal O}=[\vec{e}_1,\vec{e}_2,\cdots,\vec{e}_n]$  be an ordered basis of $\R^n$, where $\vec{e}_j \ ,j=1\cdots,n$ are the standard basis vectors of $\R^n$. For any $j$ such that $2\le j\le n-1$, there is a rearrangement of ${\mathcal O}$ that is an ordered basis $\widetilde{{\mathcal O}}=[\vec{e}_{i_1},\vec{e}_{i_2},\cdots,\vec{e}_{i_n}]$ such that, under $\widetilde{{\mathcal O}}$, $A_j(n)$ has the following block matrix representation, with blocks of either the form \eqref{eq:Ajn} above, denoted $A_1(k_i)$, or completely $0$, denoted $O$.
 \begin{align}\label{eq:AjnBlock}
     \widetilde{A}_j(n)=\begin{pmatrix}
     A_1(k_1) & O & \cdots & O\\
     O & A_1(k_2) & \ddots & \vdots \\
     \vdots & \ddots & \ddots & O \\
     O & \cdots & O & A_1(k_j)
     \end{pmatrix}
 \end{align}
 where $k_1+k_2+\cdots k_j=n$.
\end{lem}

\begin{proof}[Proof of Lemma \ref{lem:basis}]
We use an explicit construction as our method of proof.  Consider the matrix $A_j(n)$ as seen in Lemma \ref{lem:offdiagonal}.  To transform this matrix into one of the form $\widetilde{A}_j(n)$ (as seen in \eqref{eq:AjnBlock}), we need to perform a change of coordinates that involves permuting the order of the standard basis vectors.  To pick a specific ordering, we perform the following algorithm:
\begin{enumerate}
    \item Start with the standard basis vector $e_1$.  Add $j$ to the subscript.  If $j+1 \le n$, continue to add $j$ until $kj+1>n$ for some $k\in \mathbb{Z}$.  Group all of the basis vectors $e_1,e_{j+1},\ldots ,e_{(k-1)j+1}$.  Let the set $ \{ e_1,e_{j+1},\ldots ,e_{(k-1)j+1} \}$ be called $ [e_1]$. The order of the elements of $[e_1]$ and all future $[e_i]$ will be important.
    \item Repeat this process with $e_2$.  Add $j$ to the subscript.  If $j+2 \le n$, continue to add $j$ until $kj+2>n$ for some $k\in \mathbb{Z}$.  Group all of the basis vectors $e_2,e_{j+2},\ldots ,e_{(k-1)j+2}$.  Let the set $ \{ e_2,e_{j+2},\ldots ,e_{(k-1)j+2} \}$ be called $ [e_2]$.
    \item Continue this process until $e_j$ is reached.  Again, add $j$ to the subscript.  If $2j \le n$, continue to add $j$ until $kj>n$ for some $k\in \mathbb{Z}$.  Group all of the basis vectors $e_j,e_{2j},\ldots ,e_{(k-1)j}$.  Let the set $ \{ e_j,e_{2j},\ldots ,e_{(k-1)j} \}$ be called $ [e_j]$.  At this point, all of the standard basis vectors are an element of some $[e_i]$ for $i=1,\ldots ,j$.  Now, the process can be stopped.
    \item Rewrite the matrix $A_j(n)$ in the coordinate system $[ [e_1],\ldots, [e_j] ]$.  The order of the elements does matter.  Call this matrix $\widetilde{A}_j(n)$.  
\end{enumerate}
Note that this algorithm can be applied to any $A_j(n)$ and that $\widetilde{A}_j(n)$ will be a matrix of the form seen in \eqref{eq:AjnBlock}. This is illustrated by considering $A_{1j}(n)$, a matrix under the basis system $[e_1]$ along with the remaining basis elements, unaltered---excepting any perturbations that may have occurred by grouping $[e_1]$. In other words, the elements of $[e_1]$ come first and then the remaining elements follow after in the same order as they were originally. Thus, $A_{1j}(n)$ is matrix $A_j(n)$ after performing the above algorithm for only one of the potentially many basis groupings, and then changing $A_j(n)$ to representation under the new basis $[[e_1],\mbox{the rest}]$. We know $[e_1]$ is self-contained with respect to matrix multiplication by ordered $[e_1]$ elements against $A_{1j}(n)$---they will not leave the set or interfere with any other outside basis multiplications. Further and more generally, one notes that the only non-zero elements of matrices of the form $A_j(n)$ are at positions $(k,l)$, where $| k-l | = j$. By construction, both $e_k$ and $e_l$, where $k$ and $l$ satisfy the previous equality, will be in the same block $A_1(k_i)$ after the above algorithm and change of coordinates are employed. Looking at the alternate situation, positions $(k,l)$ of $A_j(n)$ where $| k-l | \ne j$, we have zero entries. This means that these basis elements do not interact via $A_j(n)$. Therefore, all elements resulting from interactions between basis vectors in different blocks of $\widetilde{A}_j(n)$ will be zero. We also know that within the grouping given by the algorithm, basis elements are shifted strictly to ``adjacent" elements under matrix multiplication. This suggests a matrix of the following form
\begin{align*}
A_{1j}(n)=\begin{pmatrix}
A_1(k_1) & O \\
O & B\\
\end{pmatrix}
\end{align*}
Where $A_1(k_1)$ is of the form \eqref{eq:Ajn} given above and $B$ represents the part of matrix $A_{1j}(n)$ not yet reordered---which may not exist at all in certain cases. The subsequent basis groups follow in the same manner, culminating in $\widetilde{A}_j(n)$. This completes the proof.
\end{proof}

We can now proceed to prove Lemmas \ref{lem:hardH0} and \ref{lem:offdiagonal}.

\begin{proof}[Proof of Lemma \ref{lem:offdiagonal}]
We see by Lemma \ref{lem:basis} that all matrices of the form $A_j(n)$ are similar---via a change of coordinates---to a matrix $\widetilde{A}_j(n)$ that can colloquially be described as a ``first off-diagonal matrix---with one entries---with potentially several one entries missing." We know via our work on the proof of Lemma \ref{lem:norm} that the norm of a matrix consisting of the two first off-diagonals is less than or equal to two. We can make this inequality strict by considering this ``almost" first off-diagonal matrix as a finite portion of an infinite dimensional lattice.

Consider the $n\times n$ matrix $A_1(n)$ consisting of the complete first off-diagonals.
$$
\begin{pmatrix}
0 & 1 & 0 & \cdots & 0 \\
1 & 0 & \ddots & \ddots & \vdots \\
0 & \ddots & \ddots & \ddots & 0 \\
\vdots & \ddots & \ddots & 0 & 1 \\
0 & \cdots & 0 & 1 & 0
\end{pmatrix}
$$
As the matrix is self-adjoint, we know that all eigenvalues are less than or equal to two---due to the bound on the norm given as part of our proof of Lemma \ref{lem:norm}. To make this inequality strict,\footnote{Actually, it is well known that all the $n$ eigenvalues of $A_1(n)$ can be computed explicitly using the $n$-th Chebyshev polynomial of the second kind, see e.g. \cite{CY,wiki}. The explicit expression shows that all the eigenvalues of  $A_1(n)$ are strictly in between $-2$ and $2$, for any $n$. Here, instead of using the explicit expression, we present a self-consistent proof for this fact.} we will contradict the following difference equation.
$$H_0 \vec{x} = 2 \vec{x}$$
Which is the same as
$$x_{k-1} + x_{k+1} = 2x_k, \ k = 1, \cdots,n$$
Where $\vec{x} = (x_1,\cdots,x_n)$ is part of an infinite system 
\begin{align}\label{eq:inf}
    x_{k-1} + x_{k+1} = 2x_k, \ k\in\Z
\end{align}
with a zero boundary condition so that $x_0 = x_{n+1} = 0$. We know that the only solutions to the difference equation above come from the fundamental set of solutions formed by the following two basis elements (i.e., these two solutions are the only ones we must check).
$$
\vec{\alpha} = \vec{1},
\ \
\vec{\gamma} = 
\begin{pmatrix}
1, \cdots, n
\end{pmatrix}
$$
Note that each of these solutions technically solves the infinite system identified in \eqref{eq:inf}, but we have used our zero boundary condition to examine a finite subsystem, specifically, one that arises from the $n \times n$ version of the above matrix. First, let us try solution $\vec{\alpha}$.
$$
\begin{pmatrix}
0 & 1 & 0 & \cdots & 0 \\
1 & 0 & \ddots & \ddots & \vdots \\
0 & \ddots & \ddots & \ddots & 0 \\
\vdots & \ddots & \ddots & 0 & 1 \\
0 & \cdots, & 0 & 1 & 0
\end{pmatrix}
\begin{pmatrix}
1 \\
\vdots \\
\vdots \\
\vdots \\
1
\end{pmatrix} =2 \begin{pmatrix}
1 \\
\vdots \\
\vdots \\
\vdots \\
1
\end{pmatrix}
$$
Comparing the first entries in the left and right hand side leads to the assertion that $1=2$.
Thus, we have a contradiction as only the trivial solution where we multiply $\vec{\alpha}$ by $0$ solves the above equation at the first place. Now, let us examine solution $\vec{\gamma}$.
$$
\begin{pmatrix}
0 & 1 & 0 & \cdots & 0 \\
1 & 0 & \ddots & \ddots & \vdots \\
0 & \ddots & \ddots & \ddots & 0 \\
\vdots & \ddots & \ddots & 0 & 1 \\
0 & \cdots, & 0 & 1 & 0
\end{pmatrix}
\begin{pmatrix}
1 \\
\vdots \\
\vdots \\
\vdots \\
n
\end{pmatrix} =
2
\begin{pmatrix}
1 \\
\vdots \\
\vdots \\
\vdots \\
n
\end{pmatrix}
$$
Examine the nth place
$$n - 1 = 2n.$$
Once again, we have a contradiction, as only $n = -1$ solves the equation at the nth place. This is impossible as the vector $\vec{\gamma}$ is structured so as to not allow this. Thus, our setup and assumptions in constructing the initial difference equation must have been incorrect due to contradiction, and $2$ cannot be an eigenvalue of $H_0$. The fact that $-2$ cannot be an eigenvalue follows from the same argument.

Now, the ``almost" first off-diagonal matrix that our reordered basis matrix $\widetilde{A}_j(n)$---provided by Lemma \ref{lem:basis}---can be shown by direct computation to have a norm less than or equal to the norm of a full first off-diagonal matrix, and thus this matrix, too, has a norm less than two. Further, $\widetilde{A}_j(n)=UA_j(n)U^T$ where $U$ is an orthogonal matrix whose columns form a permutation of the standard basis. Therefore, $\|\widetilde{A}_j(n)\|=\|A_j(n)\|$. We are thus given Lemma \ref{lem:offdiagonal}.
\end{proof}
\begin{proof}[Proof of Lemma \ref{lem:hardH0}]
We can show the desired result simply by applying Lemma \ref{lem:offdiagonal} to each off-diagonal set. Letting $A_j(n), \ j = 1,\cdots,n-1$ be given as in \eqref{eq:Ajn}, we have
\begin{eqnarray*}
    \|H_0\| & = & \|a_1 A_1(n) + \cdots + a_{n-1} A_{n-1}(n)\| \\
    \|H_0\| & \le & a_1\|A_1(n)\| + \cdots + a_{n-1}\| A_{n-1}(n)\|\\
    \|H_0\| & < & 2\sum_{j=1}^{n-1}a_j
\end{eqnarray*}
The second statement is obtained by the triangle inequality, the last is reached by applying Lemma \ref{lem:offdiagonal} to all $A_{j}(n)$ and rewriting the result as a summation.
\end{proof}

Lastly, we will complete our paper with the proof of Theorem \ref{thm:genLST2}:

\begin{proof}[Proof of Theorem \ref{thm:genLST2}]
We simply need to show that our matrix satisfies Lemma \ref{lem:power}. Invertibility of $H$ and positivity of the inverse and
Landscape Function will follow from there via Lemmas \ref{lem:inverse} and \ref{lem:posi-Green}, which are---as in Theorem \ref{thm:genLST}---completely proven by Lemma \ref{lem:power}. To show the desired result, we note that, given the conditions in Theorem \ref{thm:genLST2}
\begin{align*}
    \|V^{-1}\|\le \max_{1\le j\le n} v_j^{-1} \le \frac{1}{2\sum_{i-1}^{n-1}a_i}
\end{align*}
and, as shown by Lemma \ref{lem:hardH0}
\begin{align*}
    \|H_0\| < 2\sum_{i=1}^{n-1}a_i
\end{align*}
Thus, we see that
\begin{align*}
    \|V^{-1}H_0\| \le \|V^{-1}\| \|H_0\| < \frac{2\sum_{i-1}^{n-1}a_i}{2\sum_{i-1}^{n-1}a_i} = 1
\end{align*}
Hence, we meet the conditions to satisfy Lemma \ref{lem:power} and have invertibility and positivity of our matrix $H$ under our extended Theorem.
\end{proof}

\section*{Conclusion}
We have extended the findings of Lyra, Mayboroda, and Filoche \cite{LMF15} by proving the invertibility and positivity of a more general ``Long Range" Schr\"odinger Matrix. These findings may have physical applications because they extend the Landscape Theory to other currently in-use lattice operator models (such as the next nearest neighbor hopping extended Harper model). For further applications and study, curious readers are directed to the references section.

Our method of proof employed the power series expansion of matrices, with said method naturally extending the original results. We observe that none of our methods use any properties inherent to one dimensional spaces, suggesting that others can easily extend the findings of Landscape Theory to higher dimensional lattices and operators in the same manner.

\section*{Acknowledgements}
We would like to thank the \textit{Pi Mu Epsilon Journal} for their consideration of our submission. In addition, we would like to thank the anonymous referee for several suggestions which allowed us to improve our paper.

Further, we would like to thank Dr. Ilya Kachkovskiy and Dr. Shiwen Zhang for their guidance throughout this paper. This project was completed as a part of the Discovering America Exchange Research Program, and we would also like to thank Dr. Jeanne Wald for her work in organizing this program. The Discovering America Program helped bring together not only mathematical ideas, but people as well, and we are grateful for the opportunity to participate. The research was partially supported by NSF DMS-1758326 and DMS-1846114 grants.

\section*{About the Authors}
All students attend or have attended Michigan State University, East Lansing, Michigan, 48824. John Buhl, Yikang Li, and Xingyan Liu did not wish to provide full bios, but can be contacted at (\href{mailto:buhljhon@msu.edu}{buhljhon@msu.edu}), (\href{mailto:liyikang@msu.edu}{liyikang@msu.edu}), and (\href{mailto:cincy_leo515@163.com}{cincy\_leo515@163.com}), respectively.
\subsection*{Isaac Cinzori}
Isaac is an avid mathematics and history double major in his sophomore year at Michigan State University. After attaining his Advanced Track Mathematics degree, he hopes to use his current and future research experience while employed as a professor of mathematics. In his spare time, he enjoys reading, biking, hiking, and playing games. (\href{mailto:cinzorii@msu.edu}{cinzorii@msu.edu})
\subsection*{Isabella Ginnett}
Isabella is a physics and mathematics double major in her sophomore year at Michigan State University.  After she completes her undergraduate degree, she hopes to use her mathematical skills and research experience to obtain a PhD and pursue a career of research in high energy physics.  Outside of mathematics and physics, she enjoys to play soccer, travel, and watch movies. (\href{mailto:ginnetti@msu.edu}{ginnetti@msu.edu}) 
\subsection*{Mark Landry}
Mark is a mathematics and statistics double major in his junior year at Michigan State University. His mathematical interests are focused on probability theory, which he plans to continue to research in graduate school while obtaining his PhD. His hobbies outside of math largely include sports, mainly ice hockey, and spending time with friends. (\href{mailto:landrym5@msu.edu}{landrym5@msu.edu}) 

\begin{thebibliography}{}
\bibitem{A58} \textsc{P. W. Anderson}, \emph{Absence of diffusion in certain random lattices}, Physical Review, 109 (1958), pp. 1492--1505.

\bibitem{ADFJM} \textsc{D. N. Arnold}, \textsc{G. David}, \textsc{M. Filoche}, \textsc{D. Jerison}, and \textsc{S. Mayboroda}, \emph{Localization of eigenfunctions via an effective potential}, to appear in  Communications in Partial Differential Equations.

\bibitem{AJM}\textsc{A. Avila}, \textsc{S. Jitomirskaya}, and \textsc{C. Marx}, \emph{Spectral theory of extended Harper’s model and a question by Erd\"{o}s and Szekeres}, Invent. Math. 210 (2017), no. 1, 283--339.

\bibitem{CY} \textsc{F. Chung} and \textsc{S.-T. Yau}, \emph{Discrete Green's Functions}, Journal of Combinatorial Theory A 91, 191--214 (2000).

\bibitem{FM12} \textsc{M. Filoche} and \textsc{S. Mayboroda},\emph{Universal mechanism for Anderson and
weak localization}, Proc. Natl. Acad. Sci. USA 109 (2012), no. 37,
14761--14766.

\bibitem{F+17}
\textsc{M. Filoche}, \textsc{M. Piccardo}, \textsc{Y. Wu}, \textsc{C. Li}, \textsc{C. Weisbuch}, and \textsc{S. Mayboroda}, \emph{ Localization landscape theory of disorder in semiconductors I: Theory and modeling}, Phys. Rev. B 95, 144204, 2017.

\bibitem{LMF15} \textsc{M. L. Lyra}, \textsc{S. Mayboroda} and \textsc{M. Filoche}, \emph{Dual Landscapes in Anderson Localization on Discrete Lattices}, IOPscience and EPLjournal, vol. 109, no. 4, 2015, pp. 1--6.

\bibitem{L+16} \textsc{Lefebvre}, \textsc{A.Gondel}, \textsc{M. Dubois}, \textsc{M. Atlan}, \textsc{F. Feppon}, \textsc{A. Labb\'e}, \textsc{C. Gillot}, \textsc{A. Garelli}, \textsc{M. Ernoult}, \textsc{S. Mayboroda}, \textsc{M. Filoche}, and \textsc{P. Sebbah},  \emph{One single static measurement predicts wave localization in complex structures}, to appear in Phys. Rev. Lett..

\bibitem{S17}
\textsc{S. Steinerberger}, \emph{Localization of quantum states and landscape functions},
Proc. Amer. Math. Soc. 145 (2017), no. 7, 2895--2907.

\bibitem{WZ}
\textsc{W. Wang, S. Zhang}, \emph{The exponential decay of eigenfunctions for tight binding Hamiltonians via landscape and dual landscape functions}, in preparation.


\bibitem{wiki} \emph{Eigenvalues and eigenvectors of the second derivative}, Wikipedia.

\end{thebibliography}
\end{document}